\documentclass[10pt,journal,twocolumn]{IEEEtran}
\usepackage{amsmath,amsfonts,amsthm}
\usepackage{amssymb}
\usepackage{algorithmic}
\usepackage{algorithm}
\usepackage[caption=false,font=normalsize,labelfont=sf,textfont=sf]{subfig}
\usepackage{textcomp}
\usepackage{stfloats}
\usepackage{url}
\usepackage{verbatim}
\usepackage{graphicx}
\usepackage{cite}
\usepackage{color}
\usepackage{geometry}
\usepackage{diagbox}
\newtheorem{theorem}{Theorem}
\newtheorem{lemma}{Lemma}

\newtheorem{remark}{Remark}

\begin{document}
	
	\title{Extended AB Algorithms for Bistatic Integrated Sensing and Communications Systems}
	\author{Tian~Jiao, Yanlin~Geng,~\IEEEmembership{Member,~IEEE}, Zhiqiang~Wei,~\IEEEmembership{Member,~IEEE}, and Zai~Yang,~\IEEEmembership{Senior~Member,~IEEE}
		\thanks{T.~Jiao, Z.~Wei, and Z. Yang are with the School of Mathematics and Statistics, Xi'an Jiaotong University, Xi'an 710049, China (e-mail: tianjiao@stu.xjtu.edu.cn, zhiqiang.wei@xjtu.edu.cn,  yangzai@xjtu.edu.cn). \emph{(Corresponding author: Zai~Yang.)}} 
		\thanks{Y.~Geng is with the State Key Laboratory of ISN, Xidian University, China (e-mail: ylgeng@xidian.edu.cn).}
	}
	
	\maketitle
	
	\begin{abstract}
		Integrated sensing and communication (ISAC) is pivotal for next-generation wireless networks, rendering the computation of rate-distortion trade-off in ISAC systems critically important. In this paper, we propose the extended Arimoto-Blahut (AB) algorithms to calculate the rate-distortion trade-off in bistatic ISAC systems, which overcome the limitation of existing AB algorithms in handling non-convex constraints. Specifically, we introduce auxiliary variables to transform non-convex distortion constraints into linear constraints, prove that the reformulated linearly-constrained optimization problem maintains the same optimal solution as the original problem, and develop extended AB algorithms for both squared error and logarithmic loss distortion metrics based on the framework of AB algorithm. Numerical results validate the effectiveness of the proposed algorithm.
	\end{abstract}
	
	
	\section{Introduction} \label{SecInt}
	\IEEEPARstart{I}{ntegrated} sensing and communications (ISAC) has been widely studied since it is a key technology and research area for future wireless networks (beyond 5G and 6G). In parallel to the ISAC research centered on wireless communication applications \cite{xiao2022waveform,gaudio2020effectiveness,gao2022integrated,elbir2022rise,sankar2022beamforming}, the topic has also been the focus of recent information-theoretic research. The authors in \cite{ahmadipour2022information} examined the monostatic ISAC model from the perspective of information theory, and characterized the optimal trade-off between the capacity of reliable communication and the distortion of state estimation as an optimization problem with distortion constraints. 
	Vector Gaussian channel with in-block memory was considered in \cite{xiong2023fundamental,liu2023deterministic}, where the subspace trade-off and the random-deterministic trade-off between sensing and communication were identified. The capacity-distortion region of monostatic ISAC when the receiver have imperfect state knowledge was studied in \cite{liu2022information}. 
	The bistatic radar system, serving as a complementary paradigm to monostatic configurations, demonstrates superior channel interference mitigation capability due to its spatial diversity in transmitter-receiver separation. The authors in \cite{ahmadipour2023strong} characterized the rate-distortion and rate-detection exponent of bistatic ISAC systems, respectively, where the sensing receiver estimates or detects the state based on the known sent information. Our previous work \cite{jiao2023information} considered a bistatic ISAC system in which the sensing receiver is unaware of the sent message. The fundamental trade-off between communication rate and sensing accuracy is formulated as a rate-distortion optimization problem. In \cite{chen2025fundamental}, logarithmic loss (log-loss) function was selected to measure the quality of a soft estimate, and the corresponding capacity-distortion function of the bistatic ISAC model and the closed-form solutions for Gaussian channels under some conditions were derived. Therefore, solving the rate-distortion optimization problem is crucial for ISAC systems. However, the optimization problems for bistatic ISAC systems are particularly challenging due to their non-convex nature.
	
	The Arimoto--Blahut (AB) algorithm, developed independently by Arimoto \cite{arimoto1972algorithm} and Blahut \cite{blahut1972computation}, is a widely applied method for calculating channel capacity and rate-distortion functions in information theory. To calculate the channel capacity of point-to-point channels, the AB algorithm replaces the conditional probability mass function (pmf) by a free variable and then maximizes the objective function over each variable alternatingly.
	The authors of \cite{yasui2010toward} expanded the AB algorithm to compute the capacity region of the degraded broadcast channel, which is a non-convex optimization problem. Furthermore, the authors of \cite{liu2022blahut} developed AB-type algorithms to evaluate the supporting hyperplanes of the superposition coding region and those of the Nair-El Gamal outer bound, as well as the sum-rate of Marton’s inner bound for general broadcast channel. 
	To calculate the rate-distortion function, Blahut \cite{blahut1972computation} transformed the original distortion and compression rate problem into an unconstrained parameterized problem with respect to multipliers introduced by the distortion constraint. Then, for a fixed multiplier, a pair of distortion and rate is calculated based on a framework similar to the AB algorithm above. Finally, the multiplier is traversed to obtain the complete rate-distortion function.
	However, the operation of traversing multipliers in the AB algorithm leads to high computational complexity, making it challenging to directly determine the rate under a specified distortion, which hinders its application. Moreover, while the AB algorithm is suitable for linearly constrained optimization problems, it fails to handle those with non-convex constraints.
	
	In this paper, we develop an optimization framework for rate-distortion problems with non-convex yet differentiable distortion constraints. Our work differs from classical rate-distortion theory in two key aspects: (1) The goal of the optimization problem is to maximize communication rate rather than minimize compression rate, which arises from emerging ISAC systems; (2) The distortion constraints are non-convex with respect to the optimization variable, unlike the convex distortion in classical problems.
	
	\vspace{-0.3cm}
	\section{Problem Formulation}\label{Secpf}
	In this section, we introduce the rate-distortion trade-off problem in the bistatic ISAC system with squared error (SE) distortion and log-loss distortion, respectively. 
	
	The considered bistatic ISAC system consists of a transmitter (ISAC Tx), a communication receiver (ComRx), and a sensing receiver (SenRx). The ISAC Tx sends a codeword to convey some information to the ComRx that knows the channel states perfectly. At the same time, the SenRx at another location receives both the radiated signals of the ISAC Tx and the reflected signals from the ComRx to estimate the channel states \cite{jiao2023information}. The bistatic ISAC system was modeled as a state-dependent memoryless channel (SDMC) with two receivers, 
	where the state sequence $S^n=(S_1,\ldots,S_n)$ is independent and identically distributed generated from a given state distribution $P_S(\cdot)$ and is assumed to be perfectly and noncausally available at the ComRx but unavailable at the SenRx.
	Specifically, the transmitter encodes the message $W$ into a codeword $X^n$ and transmits it over the SDMC with two receivers. After receiving $Y^n$, the ComRx obtains the estimate of the message $\hat{W}$ by combining $Y^n$ with $S^n$. After receiving $Z^n$, the SenRx outputs $\hat{S}^n$ as the estimation of the state sequence $S^n$. 
	The performance of the decoder is measured by the average probability of error $P_{e}^{(n)}=\mathrm{Pr}\big\{\hat{W} \neq W\big\}$.
	The accuracy of the state estimation is measured by the average expected distortion $	D^{(n)}:=\frac{1}{n}\mathbb{E}\big[d(S^n, \hat{S}^n)\big]=\frac{1}{n} \sum_{i=1}^n \mathbb{E}\big[d(S_i, \hat{S}_i)\big]$,
	where $d: \mathcal{S} \times \hat{\mathcal{S}} \mapsto \mathbb{R}^{+}$ is a bounded distortion function.
	
	A pair $(R,D)$ is said to be \emph{achievable} if there exists a sequence of $\left(2^{n R}, n\right)$ codes such that $\lim _{n \rightarrow \infty} P_{e}^{(n)}=0,\varlimsup_{n \rightarrow \infty} D^{(n)} \leq D$. 
	The \emph{capacity-distortion function} is defined as $C(D)=\max\{R:\text{$R$ is achievable for the given $D$} \}$.
	\vspace{-0.3cm}
	\subsection{Bistatic ISAC System with SE Distortion}
	\label{SecBA}
	When the distortion metric is selected as SE, i.e., $d(s, \hat{s})=(s-\hat{s})^2$, we have the following result for the rate-distortion trade-off of the bistatic ISAC system.
	\begin{lemma}{\rm\cite[Theorem 6]{jiao2023information}}\label{cob2}
		The capacity-distortion function $C(D)$ for the case where $(X,S) - Y - Z$ forms a Markov chain is the optimal solution to the following optimization problem
		\begin{equation} 
			\begin{aligned}\label{op}
				\max_{P_{UX}} & \ I(X;Y|U,S)+I(U;Z)\\ 
				\mathrm{s.t.} & \ 
				\mathbb{E}\big[\big(S-\hat{S}(U,Z)\big)^2\big]\leq D,
			\end{aligned}
		\end{equation}
		where
		\begin{align*}
			\hat{S}(u, z)=\arg \min _{S^{\prime}} \sum_{s } P_{S|U Z}(s | u, z) \left(s-s^{\prime}\right)^2,
		\end{align*}
		the joint distribution of $S U X Y Z $ is given by $ P_{UX}P_{YZS|X}$ for some pmf $P_{UX}$.
	\end{lemma}
	It is observed that the objective function in the optimization problem \eqref{op} is the sum of two mutual information terms, which is non-convex with respect to the optimization variable $P_{UX}$. In addition, the constraint term in problem \eqref{op} is also non-convex with respect to the optimization variable $P_{UX}$ by observing that 
	\begin{align*}
		&\hat{S}(U=u, Z=z)\\
		=&\arg \min _{S^{\prime} } \sum_{s } P_{S|U Z}(s | u, z) \left(s-s^{\prime}\right)^2\\
		=&{\mathbb{E}}[S|U=u, Z=z]=\sum_{s } P_{S|U Z}(s|u,z)s\\
		=&\frac{\sum_{s } \sum_{x}P_{UX}(u,x)P_{ZS|X}(z,s|x)s}{\sum_{x}P_{UX}(u,x)P_{Z|X}(z|x)}, 
	\end{align*}
	which makes the problem \eqref{op} difficult to solve via existing methods.
	\vspace{-0.3cm}
	\subsection{Bistatic ISAC System with Log-loss Distortion}
	
	If the distortion metric is log-loss, i.e., $d(s, \hat{s})= -\log \hat{P}(s)$, where $\hat{S}=\hat{P}(S)$ is the soft estimator of $S$, the following result holds.
	\begin{lemma}{\rm \cite[Corollary 2]{chen2025fundamental}}\label{logloss}
		The capacity-distortion function $C(D)$ for the case where $X - Y - Z$ forms a Markov chain is the optimal solution to the following optimization problem
		\begin{equation} 
			\begin{aligned}\label{lol}
				\max_{P_{UX}} & \ I(X;Y|U)+I(U;Z)\\ 
				\mathrm{s.t.} & \ H(S|U,Z)\leq D, 
			\end{aligned}
		\end{equation}
		where the joint distribution of $S U X Y Z $ is given by $ P_{UX}P_{Y|X}P_{ZS|X}$ for some pmf $P_{UX}$.
	\end{lemma}
	It is observed that the constraint term in optimization problem \eqref{lol} remains non-convex with respect to the optimization variable $P_{UX}$, which poses a challenge to the problem solution.
	
	\section{Extended AB algorithm for bistatic ISAC}\label{SecDBAI}
	In this section, we focus on solving the optimization problems \eqref{op} and \eqref{lol}. Specifically, we show that the original optimization problem is equivalent to an optimization problem with linear constraints and develop an extended AB algorithm for the latter. 
	\vspace{-0.2cm}
	\subsection{Algorithm for Optimization Problem \eqref{op}} \label{SubSecad}
	In this subsection, we focus on solving the optimization problem \eqref{op}. Note that the non-convex constraint in optimization problem  constitutes the key obstacle to solving this problem using the existing AB algorithm. To address this issue, we introduce a new variable $c(u,z)$ that transform the non-convex constraint into a linear constraint and the corresponding optimization problem with the linear constraint is as follows
	\begin{align}\label{op2}
		\max_{p(u,x)\atop c(u,z)} & \sum p(u,x) p(y,s,z|x)\log\frac{p(x|u,y,s)p(u|z)}{p(u,x)}\notag\\  
		\mathrm{s.t.} &
		\sum p(u,x)p(z,s|x)(s-c(u,z))^2\leq D\\
		&
		\sum p(u,x)=1\notag.
	\end{align}
	For the optimization problem \eqref{op2}, we have the following result.
	\begin{theorem}\label{Th1}
		The optimal solution of the original optimization problem \eqref{op} is the same as that of the optimization problem \eqref{op2}.
	\end{theorem}
	\begin{proof}	
		The objective function in the original optimization problem \eqref{op} is as follows
		\begin{align*}
			&F\big(p(u,x)\big)=I(X;Y|U,S)+I(U;Z)\\
			=&\sum p(u,x) p(y,s,z|x)\log\frac{p(x|u,y,s)p(u|z)}{p(u,x)}.
		\end{align*}
		By comparing the original optimization problem \eqref{op} with the new optimization problem \eqref{op2}, we observe that the feasible set of the problem \eqref{op} is a subset of the feasible set of the problem \eqref{op2}. Therefore, the maximum value of the objective function in the optimization problem \eqref{op} is less than or equal to that in the problem \eqref{op2}. On the other hand, the Lagrangian function corresponding to the optimization problem \eqref{op2} is
		\begin{align*}
			&L(p(u,x),c(u,z),\lambda,\alpha)\\
			=&F\big(p(u,x)\big)+ \lambda \big(D-\sum p(u,x)p(z,s|x) (s-c(u,\\
			&z))^2\big)+ \alpha \big(\sum p(u,x)-1\big),
		\end{align*}
		where $\lambda$ and $\alpha$ are multipliers introduced for the constraints. 
		Denote the optimal solution of the optimization problem \eqref{op2} by $(p^*,c^*)$. According Karush-Kuhn-Tucker (KKT) condition, we have $L'_{c(u,z)}(p^*,c^*)=0$.
		Thus, we get 
		\begin{align*}
			\sum_{x,s} p^*(u,x)p(z,s|x) c^*(u,z)=\sum_{x,s} p^*(u,x)p(z,s|x) s,
		\end{align*}
		i.e., $c^*(u,z)=\sum_{s} p(s|u,z) s=\hat{s}(u,z)$. This implies $(p^*,c^*)$ is also a feasible solution to the original optimization problem \eqref{op}, which completes the proof.
	\end{proof}	
	Following the idea of AB algorithm, we define
	\begin{align}\label{tf}
		&\tilde{F}\big(p(u,x),q(x|u,y,s),q(u|z)\big)\notag\\
		=&\sum p(u,x) p(y,s,z|x)\log\frac{q(x|u,y,s)q(u|z)}{p(u,x)},
	\end{align}
	which is a concave function of the optimization variable $p(u,x)$ for fixed $q(x|u,y,s)$ and $q(u|z)$. Furthermore, we obtain that 
	\begin{align*}
		&F\big(p(u,x)\big)-\tilde{F}\big(p(u,x),q(x|u,y,s),q(u|z)\big)\\
		=&\sum p(u,x)p(y,s|x)D\big(p(x,y|u,s)\|\ q(x,y|u,s)\big)+\\
		&\sum p(u,z)D\big(p(u|z)\|\ q(u|z)\big)\geq   0,
	\end{align*}
	where the equality holds when $q(x|u,y,s)=p(x|u,y,s)$ and $q(u|z)=p(u|z)$. Therefore, based on Theorem \ref{Th1}, the original optimization problem \eqref{op} is equivalent to the following optimization problem 
	\begin{align}\label{opq}
		\max_{p(u,x)\atop c(u,z)} \max_{ q(\cdot|\cdot)}& \sum p(u,x) p(y,s,z|x)\log\frac{q(x|u,y,s)q(u|z)}{p(u,x)}\notag \\ 
		\mathrm{s.t.} &
		\sum p(u,x)p(z,s|x)(s-c(u,z))^2\leq D,\\
		&
		\sum p(u,x)=1,\notag
	\end{align}
	where the optimal $q(\cdot|\cdot)$ satisfies 
	\begin{align}\label{qq}
		q(x|u,y,s)=p(x|u,y,s) \quad \text{and} \quad q(u|z)=p(u|z)
	\end{align}
	and the optimal $c(u,z)$ satisfies
	\begin{align}\label{c} c(u,z)=\hat{s}(u,z)=\frac{\sum_{s,x}p_(u,x)p(z,s|x)s}{\sum_{x}p(u,x)p(z|x)}.
	\end{align} 
	The Lagrangian function corresponding to the optimization problem \eqref{opq} is
	\begin{align}\label{lagr}
		&L(p(u,x),c(u,z),q,\lambda,\alpha)\notag\\
		=&\tilde{F}\big(p(u,x),q(x|u,y,s),q(u|z)\big)+ \lambda \big(D-\sum p(u,x)\notag\\
		&p(z,s|x) (s-c(u,z))^2\big)+ \alpha \big(\sum p(u,x)-1\big),
	\end{align}
	where $\lambda$ and $\alpha$ are multipliers introduced for the constraints. 
	By setting the gradient of \eqref{lagr} with respect to the optimization variable $p(u,x)$ to zero, we get
	\begin{align} \label{oppux}
		p(u,x)= \frac{e^{d[q](u,x)-\lambda w(u,x)}}{\sum_{u,x} e^{d[q](u,x)-\lambda w(u,x)}}
	\end{align}
	where $d[q](u,x)=\sum_{y,s}p(y,s|x)\ln q(x|u,y,s)+\sum_{z}p(z|x)\ln q(u|z)$ and $w(u,x)= \sum_{s,z}p(z,s|x)\big(s-c(u,z)\big)^2$.
	Furthermore, since the capacity-distortion function of the system is monotonically non-decreasing with respect to the distortion $D$ \cite[Lemma 1]{jiao2023information}, the optimal solution of the optimization problem must be found on the boundary of the distortion constraint set. In other words, the distortion constraint is satisfied in an equality form at the optimal solution. Therefore, $\lambda$ satisfies
	\begin{align*} 
		G(\lambda) := D-\sum_{u,x}\frac{e^{d[q](u,x)-\lambda w(u,x)}}{\sum_{u,x} e^{d[q](u,x)-\lambda w(u,x)}}w(u,x)=0.
	\end{align*}
	Let $g_\lambda (u,x) = d[q](u,x)-\lambda w(u,x)$, then we get  $G'(\lambda)\geq 0$ according to Cauchy-Schwarz inequality.
	Thus, the equation $G(\lambda)=0, \lambda \geq 0$ has a unique solution when $D$ is achievable.
	Based on the derivation above, we present the proposed extended AB algorithm for optimization problem \eqref{op} in Algorithm \ref{Alg:RBA-MSE}, where the original variable $p$ and the additionally introduced variables $q$ and $c$ are updated in closed form.
	\begin{algorithm}[!h]
		\caption{Extended AB Algorithm for the Optimization Problem \eqref{op}}\label{Alg:RBA-MSE}
		\small
		\begin{algorithmic}[1]
			\REQUIRE $p(y|x,s)$, $p(z|x,s)$, $p(s)$, $p_{0}(u,x)$, $c_{0}(u,z)$, and $D$. 
			\FOR{$k=1,2,3,\dots$}
			\STATE
			Update $q_k(x|u,y,s)$ and $q_k(u|z)$ based on \eqref{qq}.
			\STATE
			Solve $\lambda_{k}$ based on the equation $G(\lambda) =0$.
			\STATE
			Update $p_k(u,x)$ based on \eqref{oppux}.
			\STATE 	
			Update $c_k(u,z)$ based on \eqref{c}.
			\ENDFOR
			\ENSURE $p(u,x)$.
		\end{algorithmic}
	\end{algorithm}
	\begin{remark}
		For the Gaussian channel model with power constraint for the bistatic ISAC system, the corresponding optimization problem with a SE distortion constraint and a power constraint is
		\begin{equation} 
			\begin{aligned}\label{opadp}
				\max_{P_{UX}} & \ I(X;Y|U,S)+I(U;Z)\\ 
				\mathrm{s.t.} & \ 
				\mathbb{E}\big[\big(S-\hat{S}(U,Z)\big)^2\big]\leq D,  \mathbb{E}[X^2] \leq B. 
			\end{aligned}
		\end{equation}
		For this optimization problem, Algorithm \ref{Alg:RBA-MSE} can be applied by replacing the Step 3 with solving a set of equations to determine the multipliers $\lambda$ and $\mu$ introduced due to the distortion and power constraints.
		In addition, the updated $p(u,x)$ in Step 4 is based on the following form
		\begin{align*}
			p_{k}(u,x)=\frac{e^{d[q_{k}](u,x)-\lambda_{k}w_{k-1}(u,x)-\mu x^2}}{\sum e^{d[q_{k}](u,x)-\lambda_{k}w_{k-1}(u,x)-\mu x^2}}.
		\end{align*}
	\end{remark}
	\subsection{Algorithm for Optimization Problem \eqref{lol}} \label{SubSeclol}
	In this subsection, we apply the proposed framework to solve the optimization problem involving a log-loss distortion constraint. 
	Following the framework discussed in the previous subsection, we replace the estimator $p(s|u,z)$ in the optimization problem \eqref{lol} by a free function $f(s,u,z)$. Then, similar to the results above, we get the optimal solution of the original optimization problem \eqref{lol} is the same as that of the optimization problem 
	\vspace{-0.3cm}
	\begin{align*}
		\max_{p(u,x),\atop f(u,s,z)} \max_{ q(\cdot|\cdot)} & \ \sum p(u,x) p(y,s,z|x)\ln\frac{q(x|u,y)q(u|z)}{p(u,x)} \\  
		\mathrm{s.t.} &\ 
		-\sum p(u,x)p(z,s|x)\ln f(u,s,z)\leq D\\
		&\ \sum p(u,x)=1, \sum_{s} f(u,s,z)=1.
	\end{align*}
	In addition, the optimal $q(\cdot|\cdot)$ satisfies $q(x|u,y)=p(x|u,y)$ and $q(u|z)=p(u|z)$, the optimal $f(u,s,z)$ satisfies $f(u,s,z)=p(s|u,z)$, and the optimal $p(u,x)$ satisfies
	\begin{align*} 
		p(u,x)= \frac{e^{d[q](u,x)+\lambda w(u,x)}}{\sum_{u,x}e^{d[q](u,x)+\lambda w(u,x)}},
	\end{align*}
	where $d[q](u,x)=\sum_{y,z}p(y,z|x)\ln q (x|u,y)q (u|z)$, $w(u,x)=\sum_{s,z}p(z,s|x)\ln f(s,u,z)$, and $\lambda$ satisfies 
	\begin{align*} 
		G_l(\lambda) :=D+\frac{\sum e^{d[q](u,x)+\lambda w(u,x)}w(u,x)}{\sum e^{d[q](u,x)+\lambda w(u,x)}}=0.
	\end{align*}
	Thus, we obtain the extended AB algorithm similar to algorithm \ref{Alg:RBA-MSE} to solve the optimization problem \eqref{lol}, the details of which are omitted due to space limitations.
	\vspace{-0.2cm}
	\subsection{Convergence Analysis} \label{SubSecCA}
	In this subsection, we primarily demonstrate the convergence of Algorithm \ref{Alg:RBA-MSE}. Let $p_{n-1}$ and $q_{n-1}$ be the values of the variables $p(u,x)$ and $\big(q (x|u,y,s),q(u|z)\big)$ at the $(n-1)$-th iteration of the algorithm, respectively.  We have the following result for the function values $\tilde{F}$ in \eqref{tf} generated in the iterations.
	\begin{theorem}\label{Th3}
		The function values $\tilde{F}$ generated in the iterations of Algorithm \ref{Alg:RBA-MSE} are monotonically non-decreasing and bounded, which satisfy $\tilde{F}(p_{n-1},$ $q_{n-1}){\leq} \tilde{F}(p_{n-1},q_{n}){\leq} \tilde{F}(p_{n},q_{n})$.
	\end{theorem}
	\begin{proof}
		For convenience, we omit the variables corresponding to the probability distribution when there is no ambiguity.
		Recalling that the definition of  $\tilde{F}(p,q)$ in \eqref{tf} and the update expression of $q$ in Theorem \ref{Th1}, we get 
		\begin{align*}
			&\tilde{F}(p_{n-1},q_{n})-\tilde{F}(p_{n-1},q_{n-1})\\
			=&\sum  p_{n-1} p(y,z,s|x)   \log\frac{p_{n-1}(x|u,y,s)p_{n-1}(u|z)}{q_{n-1}(x|u,y,s)q_{n-1}(u|z)}\\
			=&\sum  p_{n-1}(u,y,s) p_{n-1}(x|u,y,s)   \log\frac{p_{n-1}(x|u,y,s)}{q_{n-1}(x|u,y,s)}\\
			&+\sum  p_{n-1}(z) p_{n-1}(u|z) \log\frac{p_{n-1}(u|z)}{q_{n-1}(u|z)}\\
			=&\sum  p_{n-1}(u,y,s) D \big( p_{n-1}(x|u,y,s)||q_{n-1}(x|u,y,s)\big)\\
			&+\sum  p_{n-1}(z)D\big(  p_{n-1}(u|z)||q_{n-1}(u|z)\big)\\
			\geq& \ 0, 
		\end{align*}
		where $p_{n-1}(x|u,y,s)$ refers to the conditional density function generated by $p_{n-1}$, i.e., $p_{n-1}(x|u,y,s)=p_{n-1}p(y,s|x)/\sum_{x}p_{n-1}p(y,s|x)$, and the others are similar.
		
		According to the update expression of $p(u,x)$ in \eqref{oppux}, we obtain that
		\begin{align*}
			&\tilde{F}(p_{n},q_{n})-\tilde{F}(p_{n-1},q_{n})\\
			=&\sum  p_n(d[q_n]-\log p_n )-\sum  p_{n-1}(d[q_n]-\log p_{n-1} )\\
			=&\sum  p_n\lambda_n w_{n-1}+\log \sum e^{d[q_{n}]-\lambda_{n}w_{n-1}}  -\sum p_{n-1}(\\ &d[q_n]-\log p_{n-1} )\\
			=&\sum  p_n\lambda_n w_{n-1}+\sum  p_{n-1}\log \sum e^{d[q_{n}]-\lambda_{n}w_{n-1}} \\
			&-\sum  p_{n-1}(d[q_n]-\lambda_n w_{n-1}+\lambda_n w_{n-1}-\log p_{n-1} )
		\end{align*}
		\begin{align*}
			=&\sum  p_n\lambda_n w_{n-1}-\sum  p_{n-1}\log p_n -\sum  p_{n-1}\lambda_n w_{n-1}\\&+\sum  p_{n-1}\log p_{n-1} \\
			=&D(p_{n-1}||p_n)+\lambda_n\sum  p_n w_{n-1}-\lambda_n\sum  p_{n-1} w_{n-1}.
		\end{align*}
		Due to the fact that $\lambda$ satisfy $G(\lambda)=0$, we have
		\begin{align*}
			&\sum p_n w_{n-1}=\sum p_n p(z,s|x)(s-c_{n-1})^2\\
			=&\sum p_{n-1} p(z,s|x)(s-c_{n-2})^2\\
			=&\sum p_{n-1}w_{n-2}=D.
		\end{align*}
		Therefore, we get 
		\begin{align*}
			&\tilde{F}(p_{n},q_{n})-\tilde{F}(p_{n-1},q_{n})\\
			=&D(p_{n-1}||p_n)+\lambda_n\sum  p_n w_{n-1}-\lambda_n\sum  p_{n-1} w_{n-1}\\
			=&D(p_{n-1}||p_n)+\lambda_n\sum  p_{n-1}(w_{n-2}-w_{n-1}) \\
			=&D(p_{n-1}||p_n)+\lambda_n\sum  p_{n-1} p(z,s|x)(s-c_{n-2})^2\\
			&-\lambda_n\sum  p_{n-1} p(z,s|x)(s-c_{n-1})^2\\
			\geq& \ 0,
		\end{align*}
		where the last inequality holds by the fact that $c_{n-1}$ is the optimal estimator when the corresponding input distribution is $p_{n-1}$, i.e., its corresponding distortion is minimized.   
	 Furthermore, we have $\tilde{F}(p,q)\leq F(p)\leq I(X;Y|S)+I(X;Z)$, which completes the proof.
\end{proof}
	
	\section{Numerical Simulations} \label{SecNS}
	In this section, we evaluate the performance of the proposed algorithms for bistatic ISAC systems. 
	Consider an additive Gaussian channel model with a power constraint for the bistatic ISAC system. The channel from the ISAC Tx to the ComRx is  $Y=S+X+N_1$ and the channel from the ISAC Tx to the SenRx is  $Z=S+X+N_2$, where $S$, $N_1$ and $N_2$ are independently generated from $S \sim \mathcal{N} (0,\sigma_s^2),N_1 \sim \mathcal{N} (0,\sigma_1^2)$, and $N_2 \sim \mathcal{N} (0,\sigma_2^2)$.
	\begin{figure}[h]
		\vspace{-0.4cm}
		\centering
		\includegraphics[width=3in]{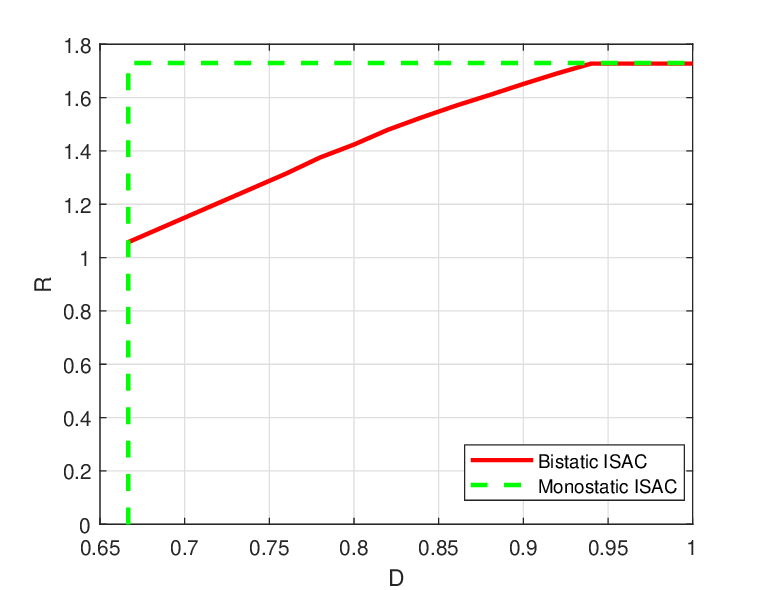}
		\vspace{-0.3cm}
		\caption{Rate-distortion functions for monostatic and bistatic ISAC systems with SE distortion constraint.}
		\label{RDAG}
		\vspace{-0.58cm}
	\end{figure}
	
	In Fig. \ref{RDAG}, we compare the rate-distortion functions of the above channel model for a monostatic ISAC system and a bistatic ISAC system, where $\sigma_s^2=\sigma_1^2=1$, $\sigma_2^2=2$ and $B=10$. 
	In a monostatic ISAC system, the estimator is aware of $X$, resulting in the distortion being $\mathrm{Var}(S|X, Z)=2/3 $ that is independent of the distribution of $X$. On the other hand, the rate is $I(X;Y|S)$, which reaches a maximum value $1/2\log_2 11=1.7297$ when $X \sim \mathcal{N} (0,10)$. As illustrated in Fig. \ref{RDAG}, the rate-distortion function for the monostatic ISAC system is located to the upper left of the rate-distortion function of the bistatic system, indicating that the system suffers from a performance loss when the estimator lacks knowledge of the sent information.
	\begin{figure}[h]
		\vspace{-0.4cm}
		\centering
		\includegraphics[width=3in]{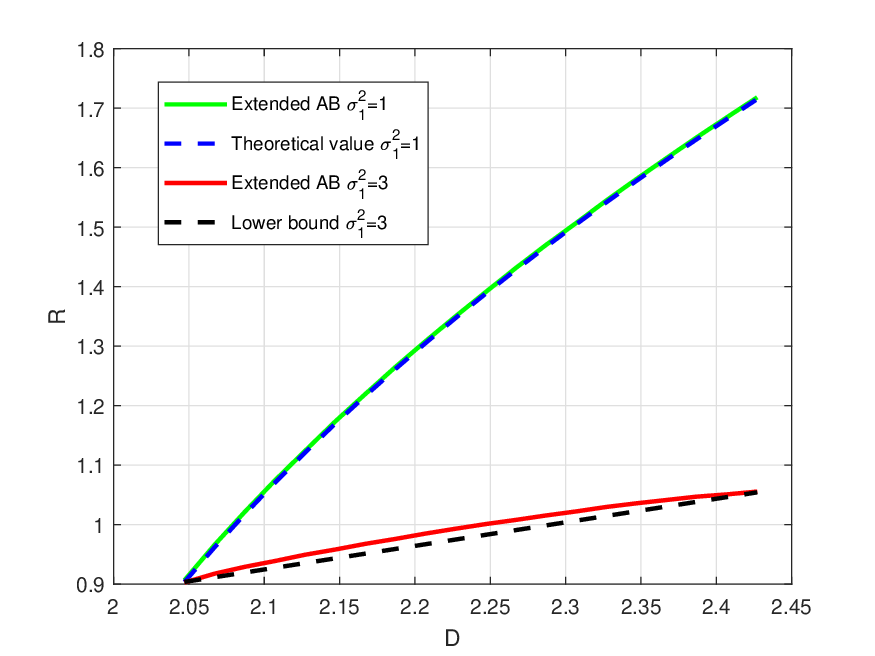}
		\vspace{-0.3cm}
		\caption{Rate-distortion function of the bistatic ISAC system with log-loss distortion constraint.}
		\label{loladd}
		\vspace{-0.3cm}
	\end{figure}
	
	Fig. \ref{loladd} depicts the rate-distortion functions of the bistatic ISAC system with log-loss distortion constraint under different parameters. Specifically, the channel parameters satisfy $\sigma_s^2=2$, $\sigma_2^2=2$, and $B=10$, and $\sigma_1^2=1<\sigma_2^2$ for the first case and $\sigma_2^2<\sigma_1^2=3<\sigma_2^2+\sigma_s^2$ for the second case. It is observed that the rate-distortion function curve obtained by the algorithm aligns well with the theoretical value for the first case, and the obtained rate-distortion function curve outperforms the theoretical lower bound \cite{chen2025fundamental} for the second case. In summary, the effectiveness of the proposed algorithm is verified.
	\vspace{-0.3cm}
	\section{Conclusions} \label{SecCon}
	In this paper, we proposed an extended AB framework to solve rate-distortion problems with non-convex constraints. Specifically, for bistatic ISAC systems, we formulated equivalent optimization problems to the rate-distortion problems for the cases with SE and log-loss distortion constraints and developed extended AB algorithms that update all variables in closed form, respectively. Additionally, we proved the convergence of the algorithm. Numerical simulations are provided to validate the effectiveness of the proposed algorithm.
	\vspace{-0.3cm}
	\bibliographystyle{IEEEtran}
	\bibliography{IEEEabrv, ISACBAref}

\end{document}